\documentclass[journal]{IEEEtran}
\usepackage{color}
\usepackage{lipsum}
\usepackage{graphicx,subcaption}
\usepackage{amsthm}
\usepackage{amssymb}
\usepackage{nomencl}
\usepackage[utf8x]{inputenc}	
\usepackage{enumerate}
\usepackage{float}
\usepackage{array}
\usepackage[labelsep=period]{caption}
\usepackage{booktabs}
\usepackage{multirow}
\usepackage{cite}
\newtheorem{theo}{Theorem}

\makeatletter
\def\hlinew#1{%
  \noalign{\ifnum0=`}\fi\hrule \@height #1 \futurelet
   \reserved@a\@xhline}
\makeatother
\ifCLASSINFOpdf
\else
\fi

\begin{document}
%
\title{Vulnerability Assessment of Load Frequency Control Considering Cyber Security}
\author{Chunyu~Chen,
        Yang Chen,
        Kaifeng Zhang,
        Wenjun Bi,
        Meng Tian 
\thanks{Chunyu Chen is with China University of Mining and Technology (email: chunyuchen@cumt.edu.cn); Yang Chen is with Nanyang Technological University (email: fedora.cy@gmail.com), Yang Chen is the corresponding author; Kaifeng Zhang and Wenjun Bi are with Southeast University, Meng Tian is with Wuhan University}
}
\maketitle

\begin{abstract}
Security is one of the biggest concern in power system operation. Recently, the emerging cyber security threats to operational functions of power systems arouse high public attention, and cybersecurity vulnerability thus become an emerging topic to evaluate compromised operational performance under cyber attack. In this paper, vulnerability of cyber security of load frequency control (LFC) system, which is the key component in energy manage system (EMS), is assessed by exploiting the system response to attacks on LFC variables/parameters.  Two types of attacks: 1) injection attack and 2) scale attack are considered for evaluation. Two evaluation criteria reflecting the damage on system stability and power generation are used to quantify system loss under cyber attacks. Through a sensitivity-based method and attack tree models, the vulnerability of different LFC components is ranked. In addition, a post-intrusion cyber attack detection scheme is proposed. Classification-based schemes using typical classification algorithms are studied and compared to identify different attack scenarios.

\end{abstract}

\subsection{Nomenclature}
\noindent \emph{1) \ Abbreviations}:
\addcontentsline{toc}{section}{Nomenclature}
\begin{IEEEdescription}[\IEEEusemathlabelsep\IEEEsetlabelwidth{$V_1,V_2,V_3$}]
\item[LFC] load frequency control.
\item[MLP] multi-perceptron .
\item[NN] neural network.
\item[SVM] support vector machine.
\item[ACE] area control error.
\item[DFT] discrete fourier transform.
\item[LSTM] long short term memory.
\item[CI] computational intelligence.
\end{IEEEdescription}
\emph{2) \ Variables}:
\addcontentsline{toc}{section}{Nomenclature}
\begin{IEEEdescription}[\IEEEusemathlabelsep\IEEEsetlabelwidth{$V_1,V_2,V_3$}]
\item[$f_0$] nominal frequency.
\item[$f_{ij}$] frequency of Generator $j$ in Area $i$.
\item[$f_i$] frequency of Area $i$.
\item[$\Delta f_i$] frequency deviation of Area $i$.
\item[$com_{ij}$]  LFC command dispatched to Generator $j$ in Area $i$.
\item[$com_{i}$] total LFC command in Area $i$.
\item[$\alpha_{ij}$] allocation coefficient of Generator $j$ in Area $i$.
\item[$\Delta f^j_i$] frequency deviation of Area $j$ interconnected with Area $i$.
\item[$\beta_{ij}$] bias coefficient of Area $j$ interconnected with Area $i$.
\item[$P^{ij}_{tie,0}$] nominal tie-line power between Area $s$ and $i$.
\item[$P^{ij}_{tie}$] tie-line power between Area $j$ and $i$.
\item[$\Delta P^{ij}_{tie}$] tie-line power deviation between Area $j$ and $i$.
\end{IEEEdescription}

\section{Introduction}
With the advent of industrial informatization in modern industrialized societies, cyber security arouses extensive attention, and becomes an emerging issue in performance evaluation of critical infrastructure, Load frequency control (LFC) serves a crucial role in system frequency stabilization, and is greatly dependent on information systems. Hence, by considering the possibility of cyber intrusion, vulnerability of cyber security of LFC should be assessed to better reflect its safety status under compromised conditions, which offers valuable information for subsequent security upgrading and reinforcement.

Diverse research studies were conducted for LFC by focusing on two aspects: 1) attack strategy \cite{esfahani2010robust,tan2017modeling,chen2018} and 2) defense strategy \cite{sridhar2014,chen2017novel}. Some researchers also consider the interaction between the attacker and defender by studying these two aspects together \cite{law2015security}.

In this paper, instead of purely analyzing the strategy and its efficacy to deteriorate (improve) LFC performance, vulnerability of cyber security of LFC system is systematically assessed. To the best of the knowledge of the authors, this specific problem has not been investigated before. Actually, the vulnerability research on other operational functions in EMS has already been conducted \cite{ten2007vulnerability,ten2008vulnerability,liu2010security,hahn2011cyber,hug2012vulnerability,liu2017power}. Ten \cite{ten2007vulnerability} used attack trees to evaluate the cyber security of supervisory control and data acquisition (SCADA) systems, and vulnerabilities were further evaluated from system, scenarios and access points \cite{ten2008vulnerability}. Vulnerability of state estimation under false data injection attack was analyzed in \cite{hug2012vulnerability}. Protection systems were considered when evaluating cyber security by simulating the physical response of power systems to malicious attacks \cite{liu2017power}.

In order to construct the assessment system, operational mechanism under cyber intrusion must be explicitly understood at first. For example, assessment for protection system oriented attacks requires clear understanding of the response of protection mechanism under attack\cite{liu2017power}. Assessment for power state estimation oriented attacks requires the knowledge of how the estimated state is falsified\cite{hug2012vulnerability}.
Apart from the mechanism, the objectives of specific function, which determine the assessment targets, should also be considered.  As for LFC, the objective is to balance the active power of the control area; hence, the degree of intentional power imbalance caused by cyber intrusion should be incorporated into vulnerability assessment. Based on the assessment targets, evaluation criteria can finally be used realize quantification of vulnerability to cyber intrusions.

According to the abovementioned description of prerequisites for assessment system construction, LFC's operational mechanism is analyzed by considering system response when the attack occurs on different variables (parameters) of LFC system. The most directly relevant indices associated with LFC, i.e., frequency and tie-line power deviation, are selected as evaluation criteria. In addition, generation disruption performance is also assessed. To this end, we analyze vulnerability by simulating intrusions into different LFC components (both from analytic and numerical analyses). Then, sensitivity-based method is adopted to quantify the vulnerability, based on which the attack tree model is used to construct the final vulnerability assessment system.

Vulnerability assessment can't mitigate the security risks but reflects the security situation of LFC system and provide apriori guidance on allocation of defense resource. Therefore, post-intrusion detection strategies are investigated as remedial countermeasures, which could promptly indicate whether or not the attack occurs, and then the defender would use the detection information to take appropriate defensive measures. Attack detection is itself an emerging trend on cyber physical system (CPS) safety and privacy. Recent advancement in computational intelligence (CI) increases its compatibility for discerning even the slightest difference, which is the foundation of detection realization\cite{soares2017cluster,jayabrindha2018ant,saha2017eeg}. Unlike natural intelligence (possessed by humans), CI is capable of analyzing complex detection (identification) problem with high accuracy. Moreover, it avoids the burdensome mathematical modelling and analytic reasoning, and is much more user-friendly. When considering the multidiversity of intrusion and the complexity of detection task, it is inevitable to adopt CI techniques in cyber intrusion detection.

In this paper, classification-based detectors are studied by considering typical CI-based classification algorithm. Specifically, multilayer perceptron (MLP), Bayesian network and support vector machine (SVM) are adopted. In order to reduce the computational (structural) complexity and high dependency on massive data, a simple yet effective dimensionality reduction method using fourier transform is applied to extract low-dimensional detection-related features.

The main contribution includes:
\begin{enumerate}
\item Vulnerability assessment for cyber security of LFC is for the first time considered. LFC performance degradation-based criteria are used to evaluate the influence of attacks. Attack tree models are then established to rank the criticality of different attack scenarios, thus laying the groundwork for subsequent protective resource allocation.
\item A post-intrusion detection scheme is presented with the aid of classification algorithms. Time-to-frequency-domain transform is applied to extract relevant features for detection, thus reducing computational complexity and enhancing detection efficiency.
\end{enumerate}

The remaining of the paper is as follows: Section \ref{sec_ba} presents basic backgrounds of cyber attack on LFC system; the influence from compromising different LFC components is systematically studied in Section \ref{sec_att}. Vulnerability assessment is performed in Section \ref{sec_vul}. The two-stage defense paradigm is discussed in Section \ref{sec_mit}. Case studies are performed in Section \ref{sec_cas}.

%
\IEEEpeerreviewmaketitle

\section{Basics of Cyber Attack on Load Frequency Control}
\label{sec_ba}
Consider Area $i$ which contains $m$ generators and is interconnected with other $k$ areas, the diagram of cyber attacks on load frequency control (LFC) of Area $i$ is as shown in Fig. \ref{lfca},
\begin{figure}[htbp]
\centering
\includegraphics[width=3 in]{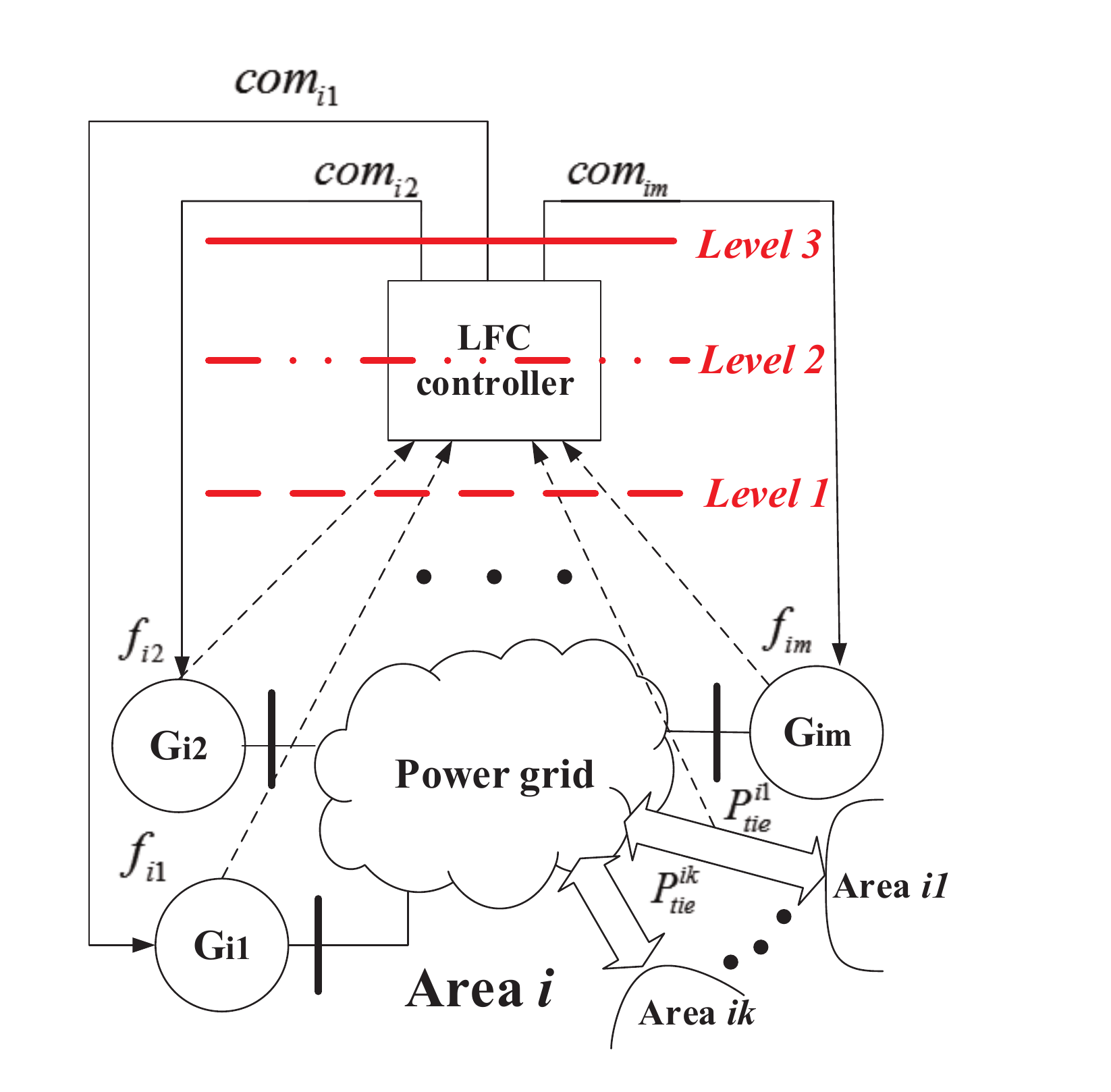}
\caption{Diagram of load frequency control}
\label{lfca}
\end{figure}
where $f_{ij}$ $1\leq j\leq m$ is the frequency measurement of generator $j$ in Area $i$; $P^{is}_{tie}$ $1\leq s\leq k$ is the interchange power of tie-line $s$; $com_{ij}$ $1\leq j\leq m$ is the LFC command dispatched to the generator $j$. Apart from the bottom-level LFC participating generators, the LFC control system can be categorized into three levels(in Fig. \ref{lfca}): Level 1(measurement upload), Level 2(LFC command generation), Level 3 (LFC order dispatch).

\subsection{Area Control Error (ACE) Control Equation}
Before introducing the types of cyber attacks on LFC, ACE control equations, which will be frequently used in studying system response (under attack) in Section \ref{sec_att}, are given as:
\begin{equation}
\begin{array}{l}
AC{E_i} = \sum\nolimits_{s = 1}^k {\Delta {P^{is}_{tie}}}  + {\beta _i}\Delta {f_i}\\
AC{E_{i1}} =  - \Delta {P^{i1}_{tie}} + {\beta _{i1}}\Delta {f^1_{i}}\\
{\kern 1pt} {\kern 1pt} {\kern 1pt} {\kern 1pt} {\kern 1pt} {\kern 1pt} {\kern 1pt} {\kern 1pt} {\kern 1pt} {\kern 1pt} {\kern 1pt} {\kern 1pt} {\kern 1pt} {\kern 1pt} {\kern 1pt} {\kern 1pt} {\kern 1pt} {\kern 1pt} {\kern 1pt} {\kern 1pt} {\kern 1pt} {\kern 1pt} {\kern 1pt} {\kern 1pt} {\kern 1pt} {\kern 1pt} {\kern 1pt} {\kern 1pt} {\kern 1pt} {\kern 1pt} {\kern 1pt} {\kern 1pt} {\kern 1pt} {\kern 1pt} {\kern 1pt}  \vdots \\
AC{E_{ik}} =  - \Delta {P^{ik}_{tie}} + {\beta _{ik}}\Delta {f^k_{i}}
\end{array}
\label{ace}
\end{equation}
where $ACE_{is}$ $1\leq s\leq k$ is ACE of Area $s$; $\Delta f^s_i$ is the frequency deviation of Area $s$; $\beta_{is}$ is the bias coefficient of Area $s$. The goal of LFC can be represented by:
\begin{equation}
\mathop {\lim }\limits_{t \to \infty } \Delta {f_i}\left( t \right) \to 0{\kern 1pt} {\kern 1pt} {\kern 1pt} {\kern 1pt} {\kern 1pt} {\kern 1pt} \mathop {\lim }\limits_{t \to \infty } \Delta P_{tie}^{is}\left( t \right) \to 0
\label{equ_ace}
\end{equation}

\subsection{Types of Cyber Attacks on Load Frequency Control}
\label{subsec_typ}
There exist various types of cyber attacks ranging from wiretapping oriented (e.g., spoofing attack) to security compromise oriented (e.g., integrity attack). In this paper the latter is considered. In respect to LFC, breach of integrity is characterized by parameter/variable falsification. Manipulation schemes of variable falsification are particularly divided as two categories:
\begin{itemize}
\item {\bf{Scale attack}}: Hackers add a gain before the true measurements:
\begin{equation}
x_a = kx_t
\label{scale}
\end{equation}
where $x_a$ is the falsified measurements; $k$ is the gain ($k=1$ represents the real measurement); $x_t$ is the true measurement.
\item {\bf{Injection attack}}: Hackers inject external disturbance signals to distort the original measurements:
\begin{equation}
x_a = x_t+d
\label{inject}
\end{equation}
\end{itemize}
\section{Influence of Cyber Attack on Three Levels of Load Frequency Control}
\label{sec_att}
Based upon the background introduction in Section \ref{sec_ba},the influence of cyber attack (as is described in Section \ref{subsec_typ}) on specific LFC components is detailed by studying the quasi-steady-state response of system frequency. LFC components are categorized into three classes based on which level (in Fig. \ref{lfca}) they belong to. As are shown in Fig. \ref{lfca1}, the components in Level 1 can be easily identified as the frequency (tie-line power measurement) $f_{is}$ ($P^{is}_{tie}$); the components in Level 3 are LFC order $com_{is}$.
\begin{figure}[htbp]
\centering
\includegraphics[width=3.5 in]{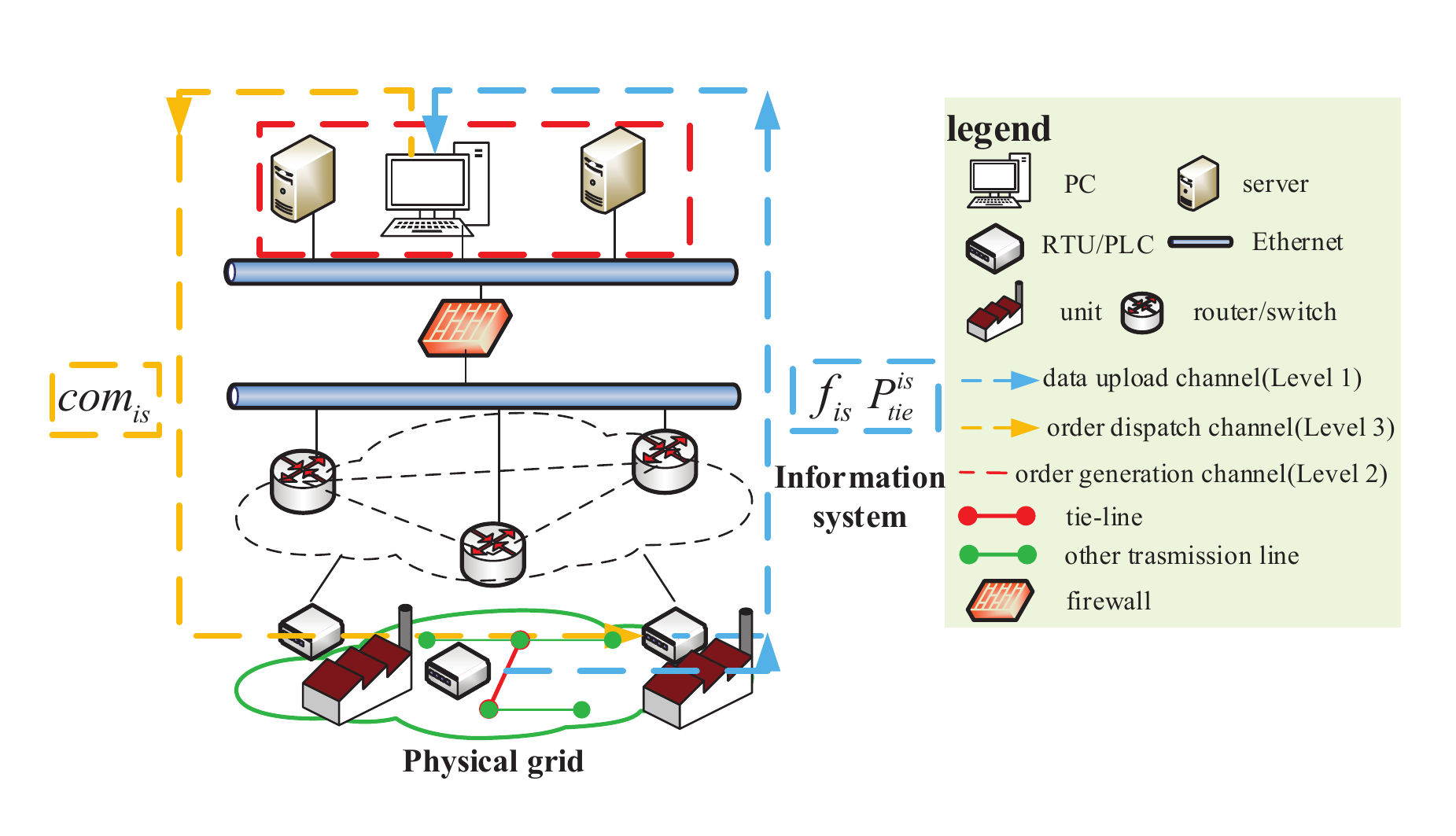}
\caption{Diagram of load frequency control}
\label{lfca1}
\end{figure}
The components in Level 2 contain intermediate variables/parameters during LFC order generation, which are shown in the red dashed box in Fig. \ref{lfcc1}
\begin{figure}[htbp]
\centering
\includegraphics[width=3 in]{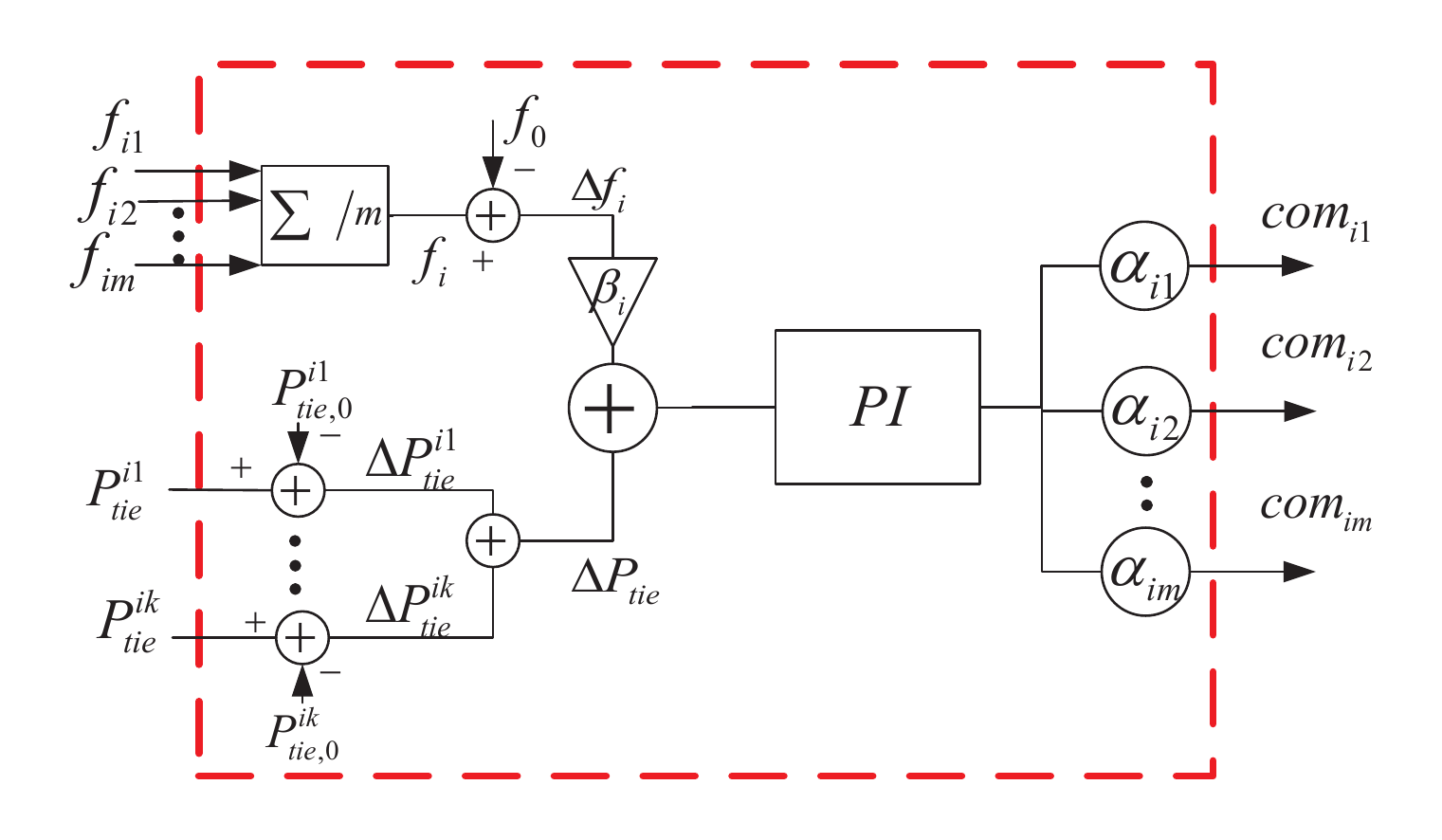}
\caption{Structure of LFC controller}
\label{lfcc1}
\end{figure}
In the remaining of this section, the quasi-steady-state response is studied by considering the attack template in (\ref{scale}) and (\ref{inject}).
\subsection{Cyber Attack on Level 1 components}
\label{lev1}
\subsubsection{$f_{is}$ oriented attack analysis}
\label{subsubf}
It can be learned that no matter the attacker adopts adopts (\ref{scale}) or (\ref{inject}) to attack frequency measurement of Generator $s$, there is a difference between the falsified center of inertia frequency $f_{ia}$ (perceived by Level 2) and the real $f_i$, which ultimately causes the miscalculation of the output feedback signal $ACE$:
\begin{equation}
AC{E_{if}} = AC{E_i} + e_i
\label{equa1}
\end{equation}
where $ACE_{if}$ is the falsified ACE; $ACE_i$ is the real ACE, ${e_i}{{ = {\beta _i}{H_{is}}\left( {k - 1} \right){f_i}} \mathord{\left/
 {\vphantom {{ = {\beta _i}{H_{is}}\left( {k - 1} \right){f_i}} {\sum\nolimits_j {{H_{ij}}} }}} \right.
 \kern-\nulldelimiterspace} {\sum\nolimits_j {{H_{ij}}} }}$ for scale attack and ${e_i}{{ = {\beta _i}{H_{is}}d} \mathord{\left/
 {\vphantom {{ = {\beta _i}{H_{is}}d} {\sum\nolimits_j {{H_{ij}}} }}} \right.
 \kern-\nulldelimiterspace} {\sum\nolimits_j {{H_{ij}}} }}$ for inject attack.

It follows that $ACE_{if}$ achieves asymptotical stability; $ACE_i+e_i=0$. It means that $\mathop {\lim }\limits_{t \to \infty } \Delta {f_i}\left( t \right) \to {{{e_i}} \mathord{\left/
 {\vphantom {{{e_i}} {\sum\nolimits_k {{\beta _k}} }}} \right.
 \kern-\nulldelimiterspace} {\sum\nolimits_k {{\beta _{ik}}} }}$. The power deviation of the $s^{th}$ tie-line interconnected with $i$ $\mathop {\lim }\limits_{t \to \infty } \Delta {P^{is}_{tie}}\left( t \right) \to  - {\beta _{is}}{{{e_i}} \mathord{\left/
 {\vphantom {{{e_i}} {\sum\nolimits_k {{\beta _{ik}}} }}} \right.
 \kern-\nulldelimiterspace} {\sum\nolimits_k {{\beta _{ik}}} }}$
\subsubsection{$P^{is}_{tie}$ oriented attack analysis}
Firstly, assume that the attacker adopts (\ref{scale}) to attack tie-line $s$, i.e., $P^{is}_{tie}=kP^{is}_{tie}$. It follows that the falsified interchange power deviation received by Area $i$ is $\Delta P^{is}_{tie}=kP^{is}_{tie}-P^{is}_{tie,0}$. However, this deviation is offset by $-\Delta P^{is}_{tie}=-(kP^{is}_{tie}-P^{is}_{tie,0})$, which is received by Area $s$. The quasi-steady-state response is not influenced. Similar conclusions can be made when  the attacker adopts (\ref{inject}).
\subsection{Cyber Attack on Level 2 Components}
\label{lev2}
As can be seen from Fig. \ref{lfcc1}, the main LFC components include $f_i$, $f_0$, $\Delta f_i$, $\beta_i$, $ P^{ij}_{tie,0}$, $\Delta P^{ij}_{tie}$, $\Delta P^{i}_{tie}$ and $\alpha_{ij}$; hence, in the remaining of this section, quasi-steady-state-response considering attack on these components is systematically analyzed.
\subsubsection{cyber attack on frequency-related components}
\label{sec1}
In this case, the attack occurs on $f_i$, $f_0$, $\Delta f_i$. Firstly, it can be learned attack on $f_i$ and $f_0$ will produce the opposite response. For brevity, only $f_i$ is considered; using (\ref{scale}) would lead to $\mathop {\lim }\limits_{t \to \infty } \Delta {f_i}\left( t \right) \to {{  {\beta _i}d} \mathord{\left/
 {\vphantom {{ - {\beta _i}d} {\sum\nolimits_k {{\beta _{ik}}} }}} \right.
 \kern-\nulldelimiterspace} {\sum\nolimits_k {{\beta _{ik}}} }}$, $\mathop {\lim }\limits_{t \to \infty } \Delta P_{tie}^{is}\left( t \right) \to {{{-\beta _{is}}{\beta _i}d} \mathord{\left/
 {\vphantom {{{\beta _{is}}{\beta _i}d} {\sum\nolimits_k {{\beta _{ik}}} }}} \right.
 \kern-\nulldelimiterspace} {\sum\nolimits_k {{\beta _{ik}}} }}$.

 As for $\Delta f_i$, it can be proved that system frequency and tie-line power still converge to nominal values unless $k$ in (\ref{scale}) satisfies
 \begin{equation}
{\beta _i}k + {\beta _{i1}} +  \cdots  + {\beta _{ik}} = 0
\label{equ_e1}
\end{equation}
Using (\ref{inject}) would lead to $\Delta f_i \to {{{\beta _i}d} \mathord{\left/
 {\vphantom {{{\beta _i}d} {\sum\nolimits_k {{\beta _{ik}}} }}} \right.
 \kern-\nulldelimiterspace} {\sum\nolimits_k {{\beta _{ik}}} }}$ ($\Delta {P^{is}_{tie}} \to  - {{{\beta _{is}}{\beta _i}d} \mathord{\left/
 {\vphantom {{{\beta _{ik}}{\beta _i}d} {\sum\nolimits_k {{\beta _{ik}}} }}} \right.
 \kern-\nulldelimiterspace} {\sum\nolimits_k {{\beta _{ik}}} }}$)

As with cyber attack in Section 1, the variables/parameters under the threat of attack are $f_0$ and $P^{is}_{tie,0}$.

Parameter modification of $f_0$ (using (\ref{inject})) will lead to $\mathop {\lim }\limits_{t \to \infty } \Delta {f_i}\left( t \right) \to {{ - {\beta _i}d} \mathord{\left/
 {\vphantom {{ - {\beta _i}d} {\sum\nolimits_k {{\beta _{ik}}} }}} \right.
 \kern-\nulldelimiterspace} {\sum\nolimits_k {{\beta _{ik}}} }}$, $\mathop {\lim }\limits_{t \to \infty } \Delta P_{tie}^{is}\left( t \right) \to {{{\beta _{is}}{\beta _i}d} \mathord{\left/
 {\vphantom {{{\beta _{is}}{\beta _i}d} {\sum\nolimits_k {{\beta _{ik}}} }}} \right.
 \kern-\nulldelimiterspace} {\sum\nolimits_k {{\beta _{ik}}} }}$. Manipulation of $P^{is}_{tie,0}$ using (\ref{inject}) will lead to $\mathop {\lim }\limits_{t \to \infty } \Delta {f_i}\left( t \right) \to {{ - d} \mathord{\left/
 {\vphantom {{ - d} {\sum\nolimits_k {{\beta _{ik}}} }}} \right.
 \kern-\nulldelimiterspace} {\sum\nolimits_k {{\beta _{ik}}} }}$, $\mathop {\lim }\limits_{t \to \infty } \Delta P_{tie}^{is}\left( t \right) \to {{{\beta _{is}}d} \mathord{\left/
 {\vphantom {{{\beta _{is}}d} {\sum\nolimits_k {{\beta _{ik}}} }}} \right.
 \kern-\nulldelimiterspace} {\sum\nolimits_k {{\beta _{ik}}} }}$.

\subsubsection{cyber attack on tie-line power-related components}
\label{sec2}
In this case, the attack occurs on $ P^{ij}_{tie,0}$, $\Delta P^{ij}_{tie}$, $\Delta P^{i}_{tie}$. Firstly,  manipulation of $P^{is}_{tie,0}$ using (\ref{inject}) will lead to $\mathop {\lim }\limits_{t \to \infty } \Delta {f_i}\left( t \right) \to {{ - d} \mathord{\left/
 {\vphantom {{ - d} {\sum\nolimits_k {{\beta _{ik}}} }}} \right.
 \kern-\nulldelimiterspace} {\sum\nolimits_k {{\beta _{ik}}} }}$, $\mathop {\lim }\limits_{t \to \infty } \Delta P_{tie}^{is}\left( t \right) \to {{{\beta _{is}}d} \mathord{\left/
 {\vphantom {{{\beta _{is}}d} {\sum\nolimits_k {{\beta _{ik}}} }}} \right.
 \kern-\nulldelimiterspace} {\sum\nolimits_k {{\beta _{ik}}} }}$.
As for scale attack on $\Delta P^{is}_{tie}$, it can be proved that system frequency and tie-line power still converge to nominal values unless $k$ in (\ref{scale}) satisfies
\begin{equation}
{\beta _{is}}(k-1) + {\beta _{i}}+{\beta _{i1}} +  \cdots  + {\beta _{ik}} = 0
\label{equ_e2}
\end{equation}
Similarly, in scale attack on $\Delta P^{i}_{tie}$  $k$ in (\ref{scale}) should satisfy
\begin{equation}
{{{\beta _i}} \mathord{\left/
 {\vphantom {{{\beta _i}} k}} \right.
 \kern-\nulldelimiterspace} k} + {\beta _{i1}} + {\beta _{i2}} \cdots  + {\beta _{ik}} = 0
 \label{equ_e3}
\end{equation}
otherwise, system frequency and tie-line power still converge to nominal values.

\subsection{Cyber Attack on Level 3}
\label{sub_cyb3}
In this scenario, $com_{ij}$ $1\leq j\leq m$ are under threat.When the total LFC order $com_{i}$ is generated through PI controller, each LFC participating generator $j$ receives power adjustment command with $com_{ij}=\alpha_{ij}com_{i}$. When the attacker changes $\alpha_{ij}$ to $\alpha_{ij,f}$, which means the falsified command is $com_{ij,f}=\alpha_{ij,f}com_{i}$. In this case, only the reference power adjustment $u_{i}$ is substituted by the sum of the real $u_i$ and the error $u_e$: ${u_{if}} = {u_i} + {u_e}$, and asymptotical stability of $ACE_i$ still holds. Moreover, since the output feedback signal is not compromised, the long-term stability is not disturbed. This extra command $\Delta co{m_{ij}} = co{m_{ij,f}} - co{m_{ij}}$ can be regarded as the feedforward compensation signal, and the steady-state $u_{if}$ remains the same as $u_i$. The only differences between $u_i$ and $u_{if}$ are transient dynamics in inception phase.

\begin{figure}[htbp]
\centering
\includegraphics[width=3 in]{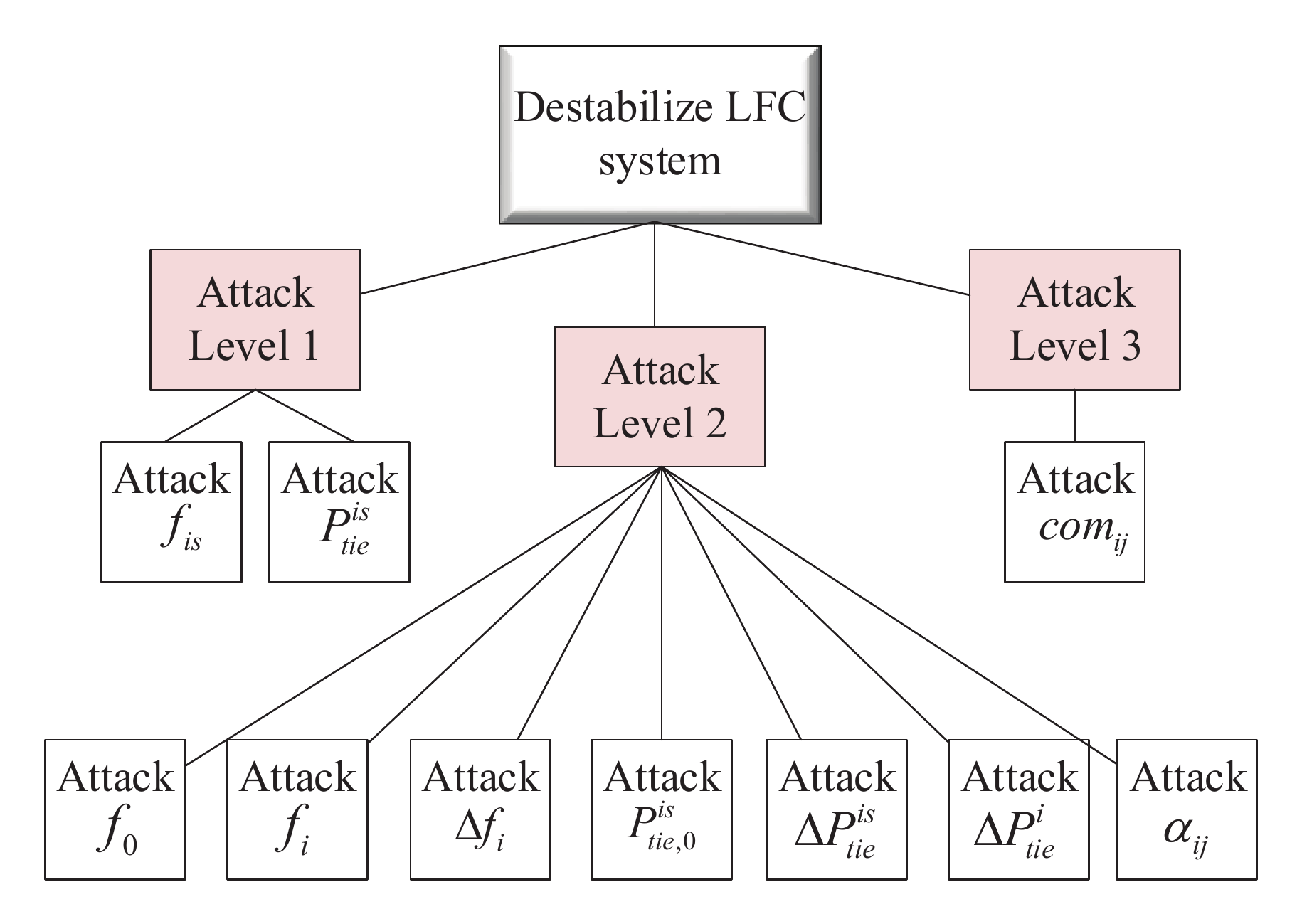}
\caption{Attack tree of LFC control system}
\label{attree}
\end{figure}
\section{Vulnerability Assessment of Cyber Security of Load Frequency Control}
\label{sec_vul}
In Section \ref{sec_att}, integrity attacks on different components of LFC control system are analyzed with respect to the influence on LFC performance. In this section, the influence is further quantified and ranked to better understand the criticality of different attack scenarios.
\subsection{Influence on Load Frequency Control Performance}
Based upon Section \ref{sec_att}, the overall laws of attack influence on LFC through compromising different components are summarized as follows:
\begin{itemize}
\item $P^{is}_{tie}$ oriented attacks on Level 1, $\beta_i$ oriented attacks on Level 2 and $com_{ij}$ oriented attacks on Level 3 and can guarantee (\ref{equ_ace}) holds.
\item  Scale attack (\ref{scale}) on $\Delta f_i$, $\Delta P^{is}_{tie}$ and $\Delta P^{i}_{tie}$ has no influence on quasi-steady-state response, unless (\ref{equ_e1}), (\ref{equ_e2}) and (\ref{equ_e3}) hold. Inject attack on $\Delta f_i$, $\Delta P^{is}_{tie}$ or $\Delta P^i_{tie}$ would cause unexpected deviations.
\item Manipulation of $f_{ij}$, $f_0$, $f_i$, $P^{is}_{tie}$ and $P^{is}_{tie,0}$ would cause unexpected deviations.
\end{itemize}
In order to quantify the influence, i.e., assign the value of the leaf nodes in Fig. \ref{attree}, the sensitivity-based method is adopted:
\begin{equation}
c_i = \frac{{{\beta _i}\left| {\Delta {f_i}} \right| + \sum\nolimits_{i = 1}^k {\left| {\Delta P_{tie}^{ik}} \right|} }}{d_i}
\label{ci0}
\end{equation}
where $d_i$ represents the integrity attack on variable/parameter $i$ in LFC; $c_i$ represents the deviation criterion under $d_i$. 

\subsection{Influence on Generation of LFC-Participating Units}
As for attacks which can cause deviations of frequency and tie-line power, e.g., $f_{ij}$ and $f_0$ oriented attacks, by manipulation of variables/parameters in LFC, ACE is falsified into :
\begin{equation}
AC{E_{if}} = AC{E_i} + er{r_i}
\label{equ_acer0}
\end{equation}
where $err_i$ is the miscalculation due to cyber attacks, e.g., $err_i$ can be $e_i$ in (\ref{equa1}). When $er{r_i} > 0$, the real ACE satisfies $AC{E_i} < 0$, which means the LFC generating units decrease generation $u_i <0$, vice versa. When considering the normal load variation $\Delta p_{di}$, the real $ACE_i$ can be represented by:
\begin{equation}
AC{E_i} = AC{E_{pi}} - er{r_i}
\label{equ_acer}
\end{equation}
where $err_i$ is the same as that in (\ref{equ_acer0}). $AC{E_{pi}}$ is the area control area induced by the real load variation $\Delta p_{di}$ ($AC{E_{pi}} =  - \Delta {p_{di}}$ when ignoring transmission loss). If $ACE_i$ in (\ref{equ_acer}) satisfies $AC{E_i} < 0$, then $u_i <0$, vice versa.

As for attacks which have no influence on long-term stability, e.g., scale attack on $\Delta f_i$, suppose that load variation $\Delta p_{di}$ occurs after the attack, which induces a frequency drop $\Delta f_i$. The falsified $\beta_{if}\Delta f_i$ will cause misadjustment of generation. Nevertheless, $u_i$ in the steady-state remains the same as before. Similarly, it can be proved that falsification of $com_{ij}$ leads to the same $u_i$ in the steady state.

By replacing the attack goal in Fig. \ref{attree} by generation disruption, a similar attack tree can be constructed. Denote value of leaf nodes by generation disruption per unit falsification:
\begin{equation}
{c_i} = \frac{\left|{\Delta {u_i}}\right|}{{{d_i}}}
\label{ci1}
\end{equation}
quantification of influence on generation by attacking different variables/parameters can be achieved.
\begin{theo}
\label{the1}
$c_i$ with respect to the same leaf node in two attack trees are equal when neglecting the transmission loss and load variation induced by the attack .
\end{theo}

\begin{proof}
Suppose under specific integrity attack $d_i$, the error of ACE induced is $e_i$. Then, the steady-state generation disruption $u_i$ is $\Delta u_i=-e_i$ when neglecting the transmission loss and load variation induced, $\left| {\Delta {u_i}} \right| = \left| {{e_i}} \right|$. On the other hand, frequency-tie-line-power deviation can be written by:
\[\begin{array}{l}
{\beta _i}\left| {\Delta {f_i}} \right| + \sum\nolimits_{i = 1}^k {\left| {\Delta P_{tie}^{ik}} \right|} \\
{\kern 1pt} {\kern 1pt} {\kern 1pt} {\kern 1pt} {\kern 1pt} {\kern 1pt} {\kern 1pt} {\kern 1pt} {\kern 1pt} {\kern 1pt} {\kern 1pt} {\kern 1pt} {\kern 1pt} {\kern 1pt} {\kern 1pt} {\kern 1pt} {\kern 1pt} {\kern 1pt}  = {\beta _i}\left| {\frac{{{e_i}}}{{\sum \beta  }}} \right| + \sum\nolimits_{i = 1}^k {\left| {\frac{{{\beta _{ik}}{e_i}}}{{\sum \beta  }}} \right|} \\
{\kern 1pt} {\kern 1pt} {\kern 1pt} {\kern 1pt} {\kern 1pt} {\kern 1pt} {\kern 1pt} {\kern 1pt} {\kern 1pt} {\kern 1pt} {\kern 1pt} {\kern 1pt} {\kern 1pt} {\kern 1pt} {\kern 1pt} {\kern 1pt} {\kern 1pt}  = \left| {{e_i}} \right|\left( {\frac{{{\beta _i}}}{{\sum \beta  }} + \sum\nolimits_{i = 1}^k {\frac{{{\beta _{ik}}}}{{\sum \beta  }}} } \right)\\
{\kern 1pt} {\kern 1pt} {\kern 1pt} {\kern 1pt} {\kern 1pt} {\kern 1pt} {\kern 1pt} {\kern 1pt} {\kern 1pt} {\kern 1pt} {\kern 1pt} {\kern 1pt} {\kern 1pt} {\kern 1pt} {\kern 1pt} {\kern 1pt} {\kern 1pt}  = \left| {{e_i}} \right|=\left| {\Delta {u_i}} \right|
\end{array}\]

\end{proof}
Based on Theorem \ref{the1}, it is learned that the two vulnerability indices, disruption of LFC performance and generation, have the same scores in respect to certain integrity attacks in Section \ref{sec_att}, which means that performance and generation disruption have a strong positive correlation. And the two indices are interchangeable when assessing vulnerability of LFC.
\section{Mitigation for Vulnerability of Cyber Security of Load Frequency Control}
\label{sec_mit}
Mitigation for vulnerability can be summarized as two stages of defense. Stage 1 is the prevention stage where the operator tries to prevent the attacker from infiltrating into the system; stage 2 is the detection stage where the operator should single out the compromised signal once the system is infiltrated, which lays the foundation for mitigation through system reconfiguration.

With the aid of the designed attack tree model in Section \ref{sec_vul}, preventive resource can be allocated based on the criticality of each attack scenario (indicated by the leaf node in the attack tree). Preventive resource is used to invalidate potential attacks before they ever infiltrate the system. For example, the operator can set up multiple sensors for the same variable based upon the presumption that the attacker cannot compromise all of them, which is also known as measurement redundancy method.

Attack detection and identification is capable of discerning attack activity when preventive measures fail to resist it, which is essential to grasping behavioral pattern of the attacker besides information provision for subsequent mitigatory measure design. Tough statistical methods can effectively solve the differentiation problem among normal load variation and cyber attack scenarios\cite{law2015security,sridhar2014}, the type of attack scenarios cannot be identified. As for cyber attack on LFC studied in this paper, the type represents the category of the attack signal. In order to identify to which category a new observation of a scenario belongs, classification-based method is adopted to solve this typical multi-class anomaly detection techniques.

Classification-based techniques operate in a two-phase fashion, where the training phase learns a classifier using labelled training data and the test phase identifies to which class a test instance belongs to by using the trained classifier. Apart from the classification algorithm (e.g., neural networks (NNs)-based and Bayesian networks-based ones), the key element of determining the classifier quality is input data, which determines the upper limit of learning algorithm.

\subsection{Input Data for Classifier}

Instead of relying on static responses (quasi-steady-state response) for data generation, the input data should be tightly connected with the LFC dynamics, which can generate dynamic system response containing abundant information about properties of scenarios. Hence, dynamic responses of area control error (ACE), which considers both frequency and tie-line power variation, are used as the input data. Every ACE instance is a time series, with the expansion of time length and increase of sampling rate, the amount of data points grows up. When considering hundreds even thousands of data instances, the total amount of data becomes even much greater and big data forms.

There exists much redundant and classification-irrelevant information in data instances, and feature selection or extraction should be implemented to obtain relevant features, thus reducing classifier complexity and enhancing generalization performance. Though deep learning techniques can self-learn multi-level representations and features for classification, it usually demands a huge quantity of data, which is impossible in reality since cyber attacks are a small probability events. In order to deal with the feature generation with limited data instances, discrete Fourier transform (DFT) is used to extract the low-dimensional components (DFT coefficients) in frequency domain. They preserve obvious differences of different classes and significantly reduce data complexity, which are quite beneficial to data training (testing) in an efficient fashion.

ACE responses under normal and compromised conditions are simulated and the results are shown in Fig. \ref{simu1}. As can be seen, compared the variation of system frequency responses (in Fig. \ref{a0}), the variation of ACEs (in Fig. \ref{a1}) is much more explicit among different conditions. Furthermore, the dimension of inputs is significantly reduced by implementing DFT on ACEs without deteriorating the explicit variation (in Fig. \ref{a2}).
\begin{figure*}
\begin{subfigure}{0.3\textwidth}
   \includegraphics[width=\linewidth]{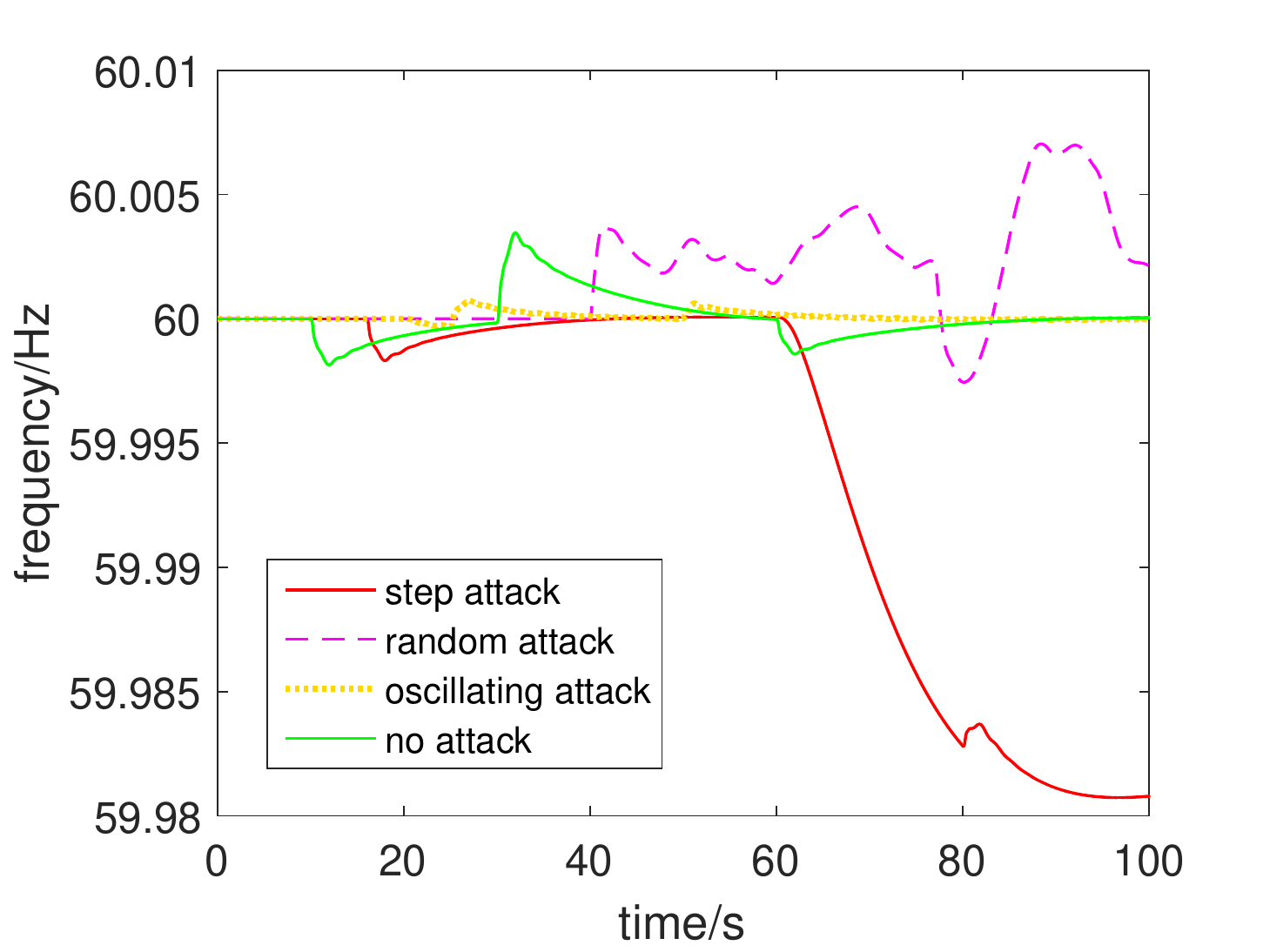}
   \caption{frequency under different conditions} \label{a0}
\end{subfigure}
\hspace*{\fill}
\begin{subfigure}{0.3\textwidth}
   \includegraphics[width=\linewidth]{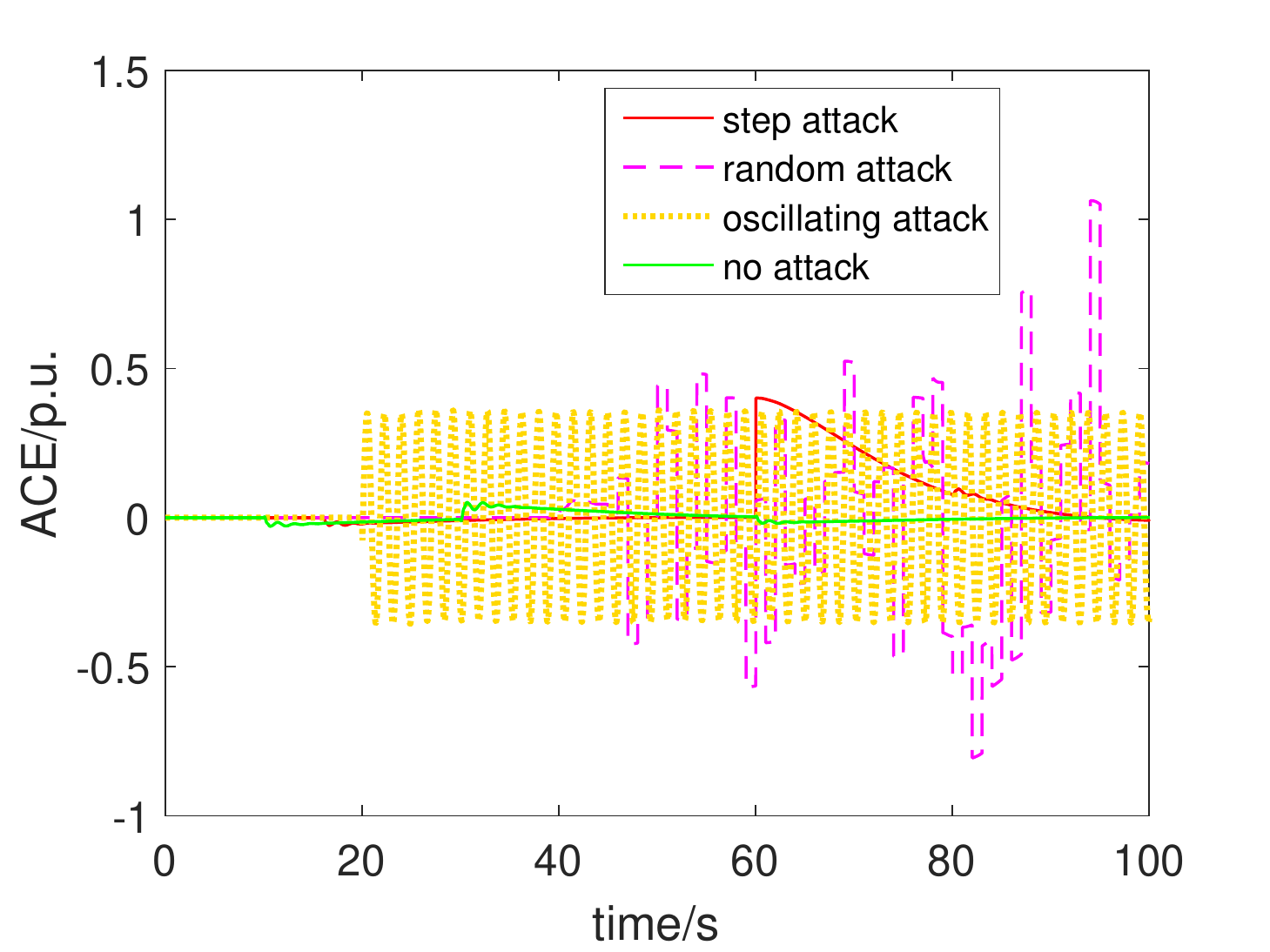}
   \caption{ACE under different conditions} \label{a1}
\end{subfigure}
\hspace*{\fill}
\begin{subfigure}{0.3\textwidth}
   \includegraphics[width=\linewidth]{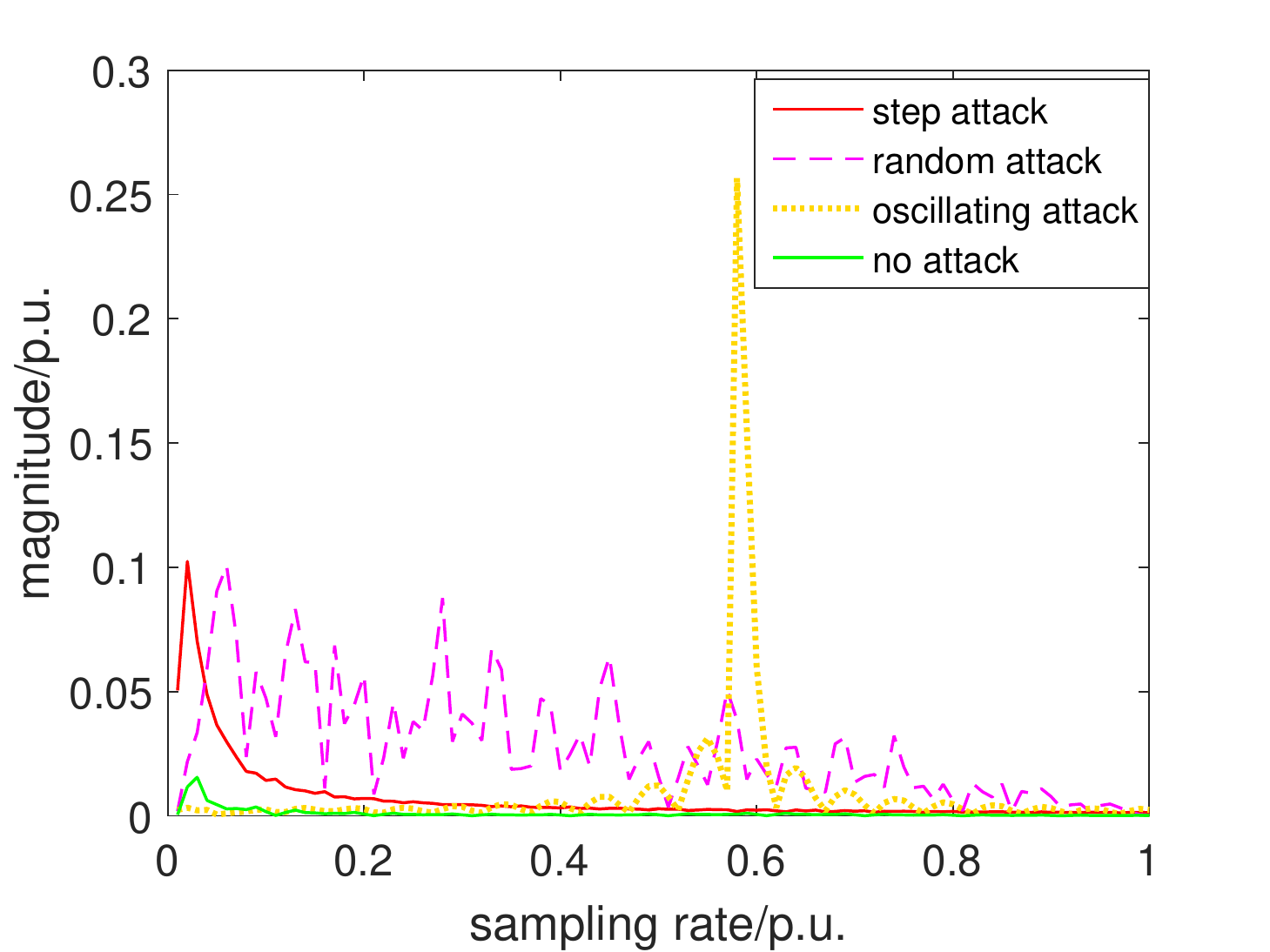}
   \caption{spectral components under different conditions} \label{a2}
\end{subfigure}
\hspace*{\fill}

\caption{LFC system responses under different conditions}
\label{simu1}
\end{figure*}

\subsection{Classification Algorithm}
\label{subsec_cla}
\subsubsection{Neural Networks}
In this paper, the feedforward neural network multilayer perceptron (MLP) is used for classification\cite{windeatt2011embedded}. The structure is as shown in Fig. \ref{mlp}.
\begin{figure}[htbp]
\centering
\includegraphics[width=3 in]{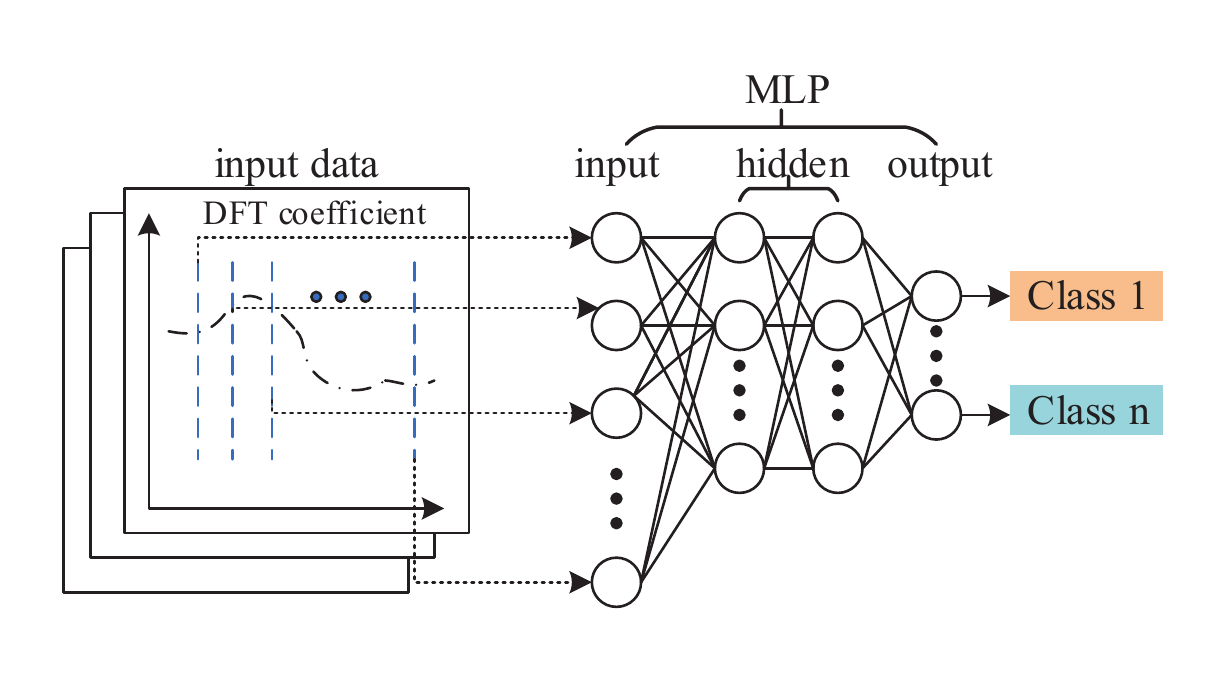}
\caption{Attack tree of LFC control system}
\label{mlp}
\end{figure}

Mathematically, the goal of MLP classifier is derived from
\begin{equation}
\min E\left( \omega  \right) = \frac{1}{2}\sum\limits_{i = 1}^p {{{\left\| {y\left( {{x^i},\omega } \right) - {d^i}} \right\|}^2}}
\end{equation}
where $\omega$ represent the weights; $E$ is the error term; is the $i^{th}$ input data, $y$ is the output of classifier, $d^i$ is the $i^{th}$ desired output. Through gradient descent approach, the optimal can be computed to minimize the prediction errors.

Besides MLP, the autoencoder is adopted for data denosing/dimensionality reduction. Dataset of DFT coefficients is firstly fed into the autoencoder\cite{liou2014autoencoder,chen2018autoencoder}, the output of which is then used as the input of MLP classifier. Moreover, in case of data indistinguishable-ness under normal load variation and some attack scenarios, a threshold-based module is developed. The architecture of the whole composite MLP classifier is shown in Fig. \ref{sche}.
\begin{figure}
\includegraphics[width=3.5in]{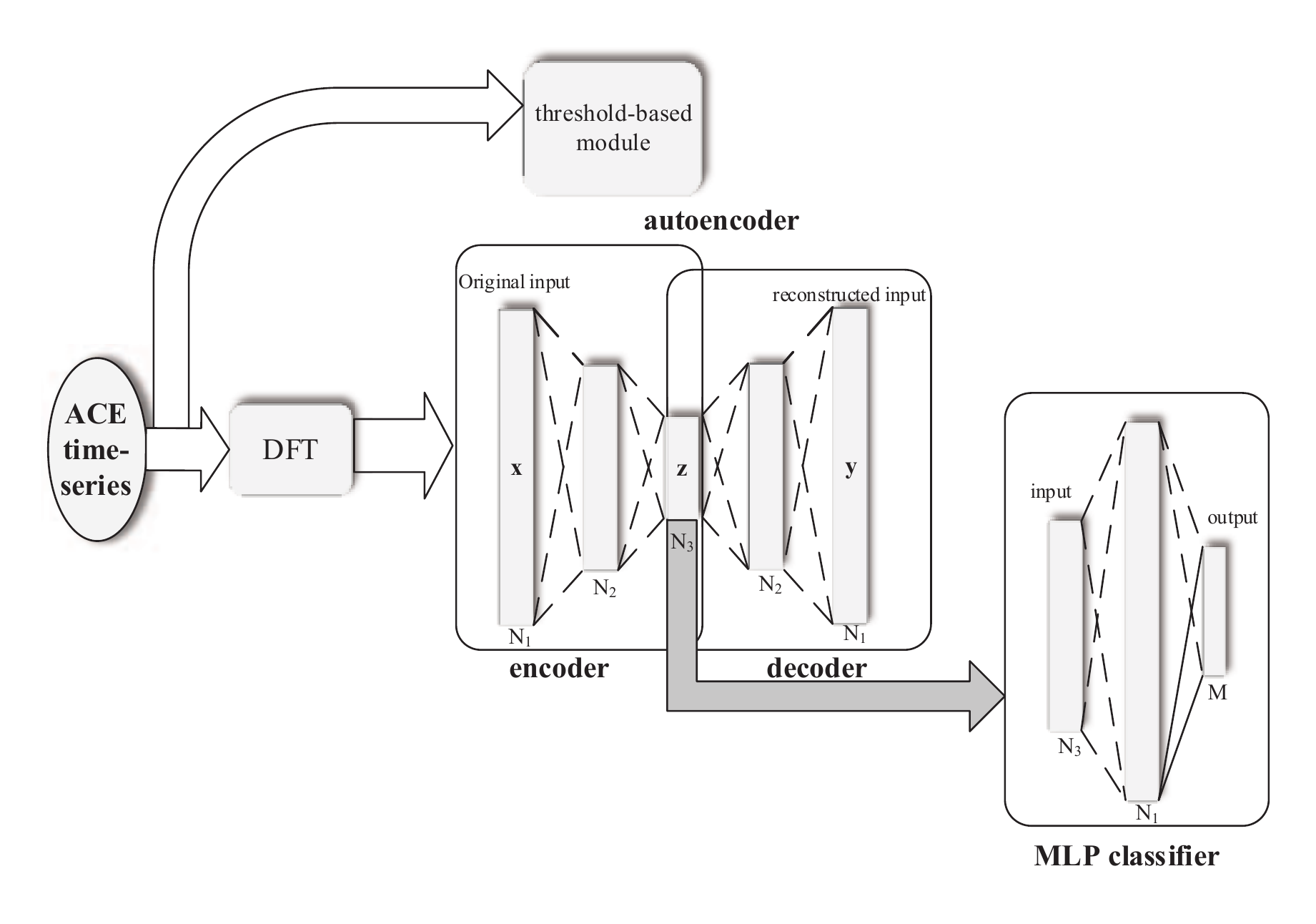}
\caption{Schematic diagram of composite NN-based attack detector} \label{sche}
\end{figure}
Details of the threshold-based module are given in Algorithm 1.
\begin{table*}
\centering
\begin{tabular}{l}
\hlinew{1pt}
{\bf{Algorithm 1}} Threshold-based detection module in Fig. \ref{sche}\\
\hline
1: {\bf{Sample}} $\Delta f_o$ (the sampling rate is $f_s$, duration of inspection is $T_s$, the length of the signals is $L_s=T_sf_s$)\\
2: {\bf{Set}} sliding window width $sw_1=l_1$        \\
3: {\bf{For}} $i=1: L_s-l_1+1$         \\
4:  $y(i)=\Delta f_o(i:i+l_1-1)$\\
5:  $x(i)=exp(var(y(i)))$ \\
6:  {\bf{End}} \\
7: {\bf{If}} $x(i)>=\varsigma_f$\\
8: {\bf{Then}} $\Delta f_o$ is compromised \\
9: {\bf{End}}\\

\hline
\end{tabular}
\end{table*}
\subsubsection{Support Vector Machine}
The goal of support vector machine classification (for muli-classification) is to divide a separable dataset into subsets, it can be defined as an optimization problem\cite{hsu2002comparison}
\begin{equation}
\begin{array}{l}
\mathop {\min }\limits_{\omega ,b,\xi } \sum\limits_{m = 1}^k {{{\left\| {{\omega _m}} \right\|}^2}}  + C\sum\limits_{i = 1}^l {\sum\limits_{m \ne {y_i}} {\xi _i^m} } \\
s.t.\omega _{{y_i}}^T\phi \left( {{x_i}} \right) + {b_{{y_i}}} \ge \omega _m^T\phi \left( {{x_i}} \right) + {b_m} + 2 - \xi _i^m\\
\xi _i^m \ge 0
\end{array}
\end{equation}
where ${\xi _i^m}$ represents slack variables; $\omega_m$ represents the weights; $C$ is the regulation parameter which reflects satisfaction degree of the constraints; $x_i$ is the input of an instance; $\phi$ is the basis function which mapping to a high dimensional space for better separability. It constructs two-class rules where $\omega _m^T\phi \left( {{x_i}} \right) + {b_m}$ separates training instances of class from other classes.
\subsubsection{Bayesian Network}
Naive Bayesian classifier \cite{koc2012network} predicts that a data instance $X$ belongs to the class with the highest a posteriori probability conditioned on $X$, which means $X$ belongs to class $C_i$ if and only if
\begin{equation}
P\left( {{{{C_i}} \mathord{\left/
 {\vphantom {{{C_i}} X}} \right.
 \kern-\nulldelimiterspace} X}} \right) > P\left( {{{{C_j}} \mathord{\left/
 {\vphantom {{{C_j}} X}} \right.
 \kern-\nulldelimiterspace} X}} \right),1 \le j \le m,j \ne i
 \end{equation}
 In order to calculate the maximum, $P\left( {{{{C_i}} \mathord{\left/
 {\vphantom {{{C_i}} X}} \right.
 \kern-\nulldelimiterspace} X}} \right)$ should be calculated first, Based on Baye's theorem, it follows that
\begin{equation}
P\left( {{{{C_i}} \mathord{\left/
 {\vphantom {{{C_i}} X}} \right.
 \kern-\nulldelimiterspace} X}} \right) = \frac{{P\left( {{X \mathord{\left/
 {\vphantom {X {{C_i}}}} \right.
 \kern-\nulldelimiterspace} {{C_i}}}} \right)P\left( {{C_i}} \right)}}{{P\left( X \right)}}
  \end{equation}
 where ${P\left( {{X}} \right)}$ is equal for all classes; ${P\left( {{C_i}} \right)}$ can be computed by counting the frequency of instance belonging to in all data; based on the naive assumption of conditional independence of each attribute (data point in DFT coefficients), one has
\begin{equation}
P\left( {{X \mathord{\left/
 {\vphantom {X {{C_i}}}} \right.
 \kern-\nulldelimiterspace} {{C_i}}}} \right) \approx \prod\limits_{k = 1}^n {P\left( {{{{x_k}} \mathord{\left/
 {\vphantom {{{x_k}} {{C_i}}}} \right.
 \kern-\nulldelimiterspace} {{C_i}}}} \right)}
   \end{equation}

\section{Case Studies}
\label{sec_cas}
In this section, Kundur's 4-unit-13-bus system is used for vulnerability assessment and subsequent attack scenario identification. The single-line diagram is as shown in Fig. \ref{lfca}.
\begin{figure}[htbp]
\centering
\includegraphics[width=3.5 in]{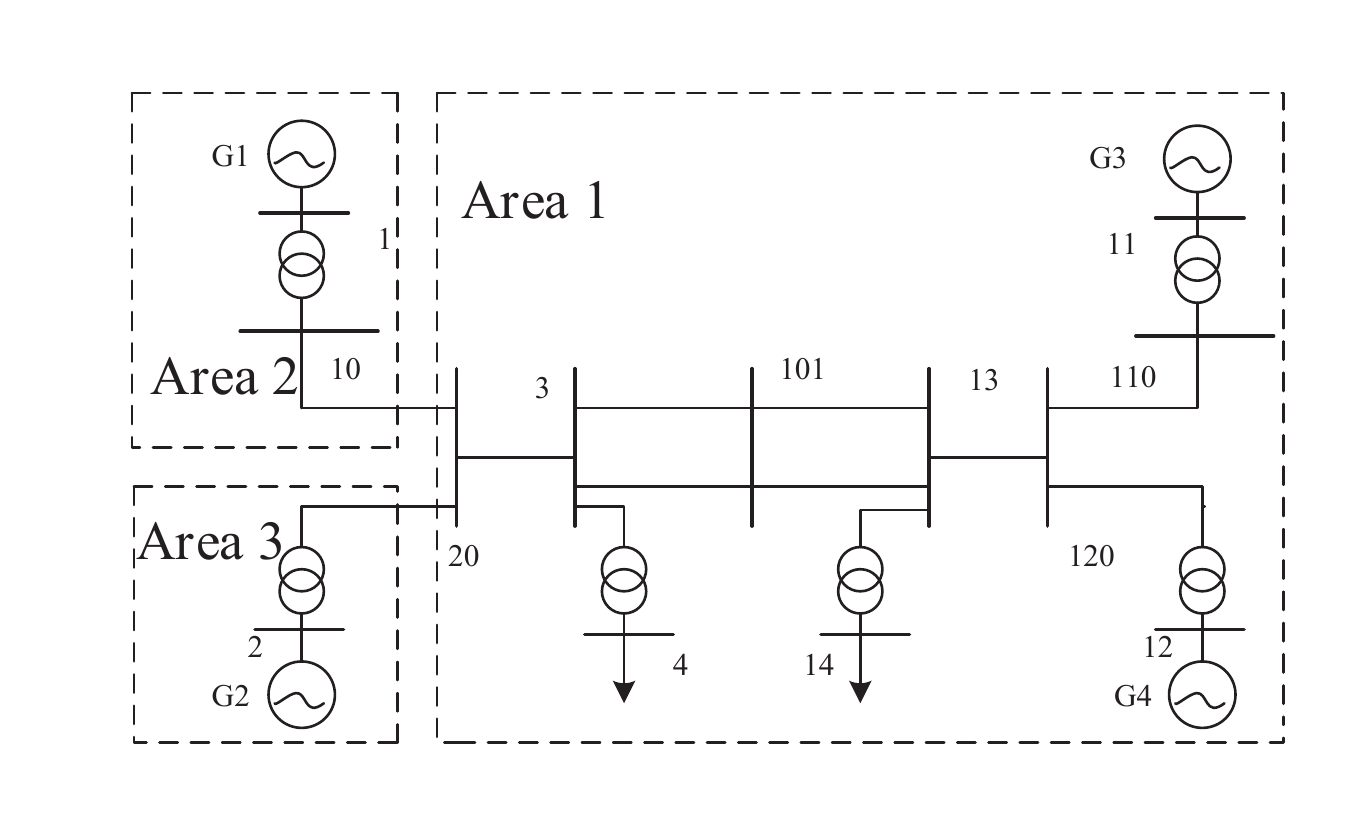}
\caption{Single line diagram of IEEE 13-bus based tree-area}
\label{three}
\end{figure}

\subsection{Vulnerability Assessment for Cyber Security of LFC}
\label{subsec_att}
The system is divided into three control areas where Area 1 is the subject for study. The attack model (leaf node in attack tree) is constructed by simulating each scenario addressed in Section. 25 total attack scenarios are generated and numbered. Descriptions of scenarios corresponding to specific numbers are given in Appendix \ref{appen1}.

Sensitivity index in (\ref{scale}) and (\ref{inject}) are computed respectively for each scenario. The results are shown as histograms in Fig. \ref{att_tree}
\begin{figure*}
\centering
\includegraphics[width=6.8 in]{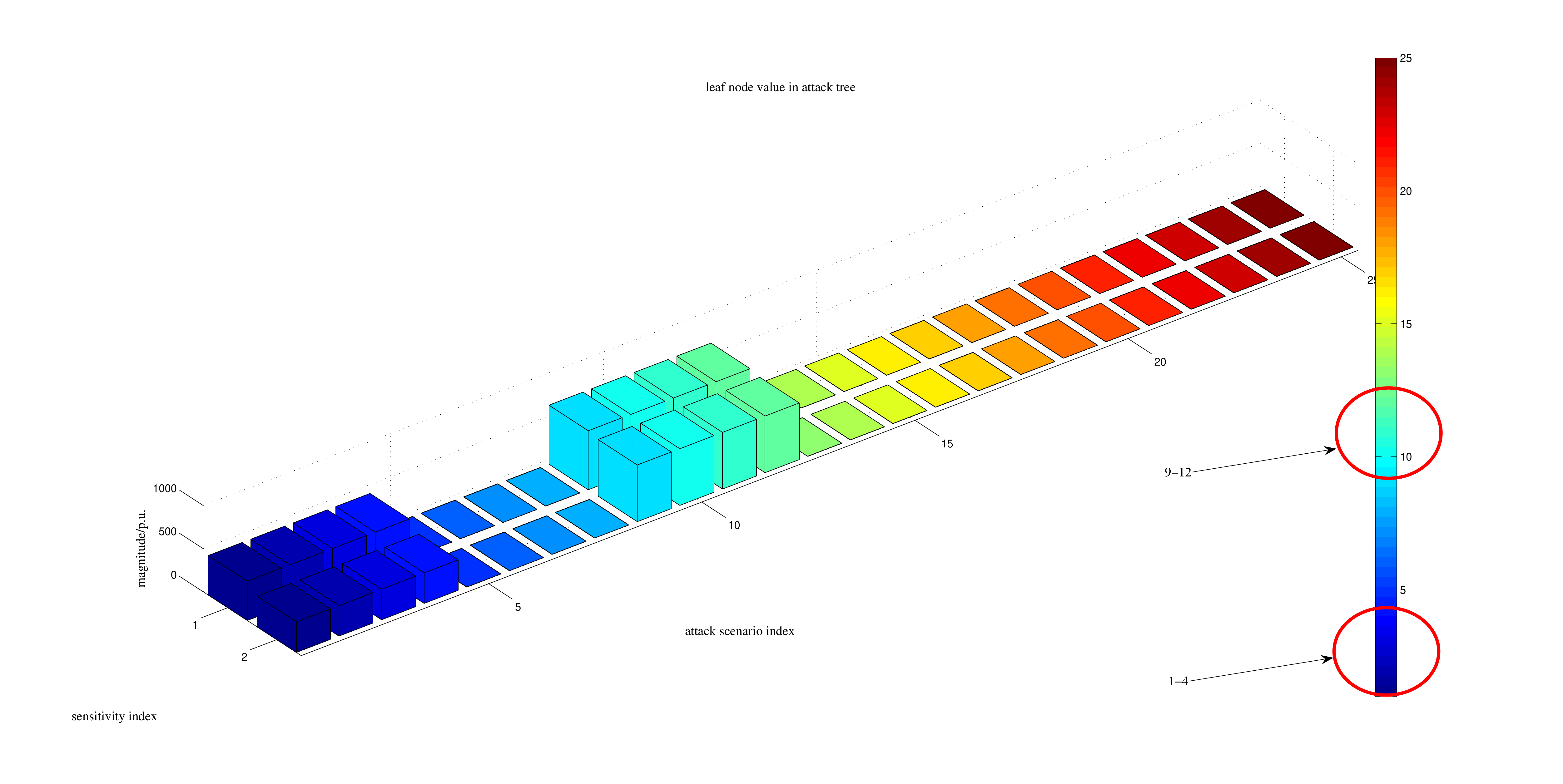}
\caption{Leaf node value of different attack scenarios in attack tree}
\label{att_tree}
\end{figure*}
where the number on $x$ axis represents the type of sensitivity, 1 means (\ref{ci0}) is used; 2 means (\ref{ci1}) is used. the number on $y$ axis represents the attack scenario index.
From Fig.\ref{att_tree} it can be learned that attacks on $f_0$, $f_1$ and $\Delta f_1$ (Scenario 9-12) produce the largest degree of disruption on both LFC performance (684) and generation (657), followed by $f_{11}$ ($f_{12}$) oriented attacks (Scenario 1-4). The difference is from the scale-down of weighted sum operator ${{{H_{is}}d} \mathord{\left/
 {\vphantom {{{H_{is}}d} {\sum\nolimits_{j = 1}^m {{H_{ij}}} }}} \right.
 \kern-\nulldelimiterspace} {\sum\nolimits_{j = 1}^m {{H_{ij}}} }}$ for the $f_{11}$ ($f_{12}$) oriented attacks, which is smaller than $d$ for $f_0$ ($\Delta f$) oriented attacks.

Nevertheless, tie-line power oriented attacks (e.g., $P^{11}_{tie,0}$ and $\Delta P^{11}_{tie}$) produce negligent disruption (1) compared with frequency oriented attacks. It stems from the amplification effect of $\beta_1$ in ACE; $\beta_1$ is the sum of the steady gain $1/R$ and damping constant $D$ of generators, which is usually several hundred to thousand.

It can also be learned that the node values for these two indices in the same scenario are not completely the same, which does not contradict with Theorem \ref{the1}. Since in Theorem \ref{the1}, it is assumed there exists no transmission loss. In practical power systems, transmission loss cannot be ignored. Moreover, when attack occurs, which is in the form of active power disturbance, the voltage profiles will also change, which further causes variation of voltage-dependent loads. However, this does not affect the interchangeableness of disruption of LFC performance and generation indices in general.
\subsection{Vulnerability Mitigation for Cyber Security of LFC}
In this section, classification-based scenario identification is simulated using Kundur's system. The simulated scenarios include: 1) normal load variation, 2) step attack, 3) random attack, 4) oscillating attack. The three classification algorithms in Section \ref{subsec_cla} are executed.

\subsubsection{Detection using Composite MLP-based Classifier}
\label{simu21}
MLP-based classifier is firstly tested. 240 ACE data instances are produced by simulating LFC on Kundur's system using MATLAB/SIMULINK (each scenario contains 60 data instances). It should be mentioned that the simulation model contains complete electromechanical dynamics and can reflect the characteristics of real system. Hence, the simulation data can be approximately used as the real-life data.

$N_1$, $N_2$ and $N_3$ in Fig. \ref{sche} are chosen as $100$, $60$ and $30$ respectively; $M$ is $3$ in this case. The detector is constructed with Keras, which is a high-level neural networks API employing TensorFlow as its backend. The whole dataset (DFT coefficients) is split into training (70\%) and test (30\%) datasets. In Fig. \ref{train}, four curves corresponding to loss and accuracy of the classifier with and without autoencoder are given. Since the number of epochs to train the model is set as 10, the first ten points on each curve quantify the performance in training epochs, and the last one quantifies the performance in testing. From Fig. \ref{train}, autoencoding assists the classifier in overfitting avoidance; and the generalization performance of the classifier is enhanced.
\begin{figure}[htbp]
\centering
\includegraphics[width=3.5 in]{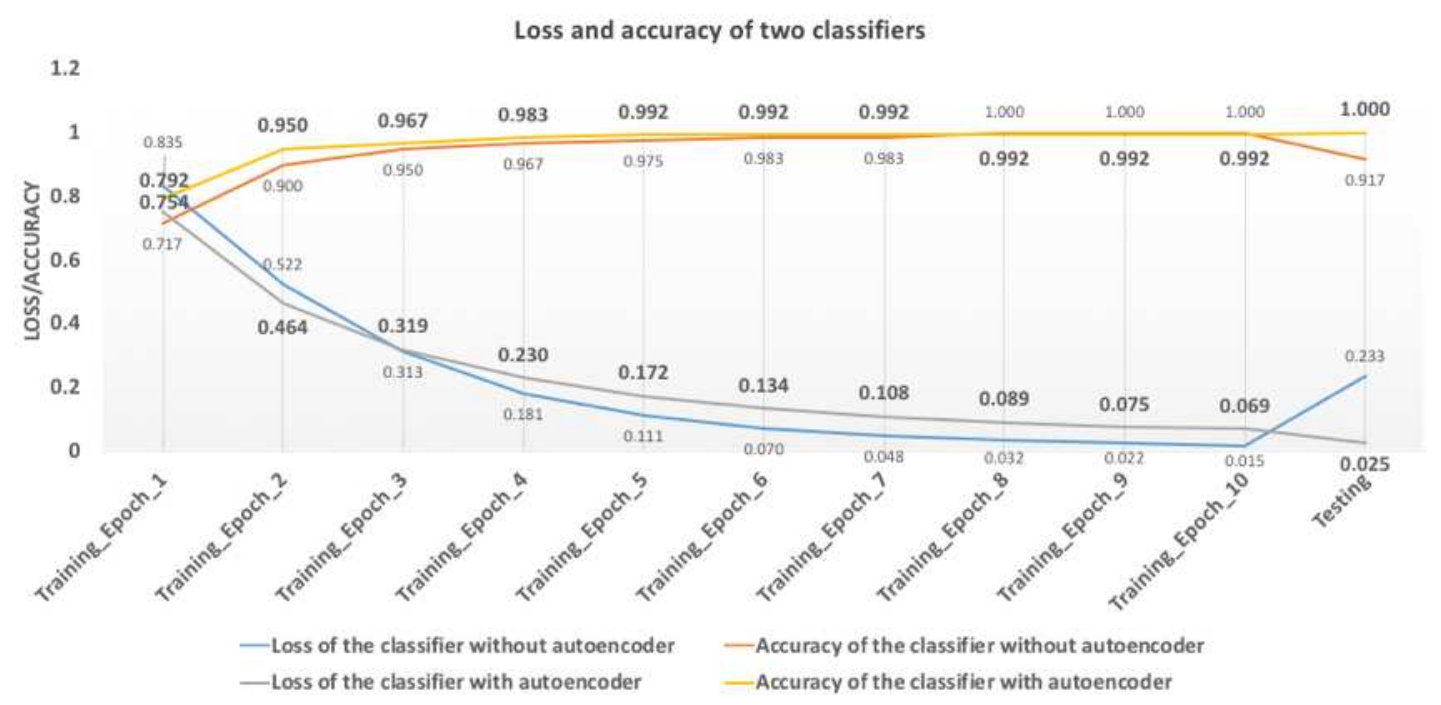}
\caption{Simulation results of MLP-based classifier}
\label{train}
\end{figure}
For comparison, the raw ACE time-series data are trained using long short time memory (LSTM) deep networks. The results are shown in Table \ref{t2}

\begin{table}[htpb]
\centering
\begin{tabular}{c c c}
\hline
 & detection accuracy & elapsed time ($s$)\\ \hline
 LSTM & 0.95 & 284.4\\
 \hline
 MLP  & 1 & 0.0023\\
 \hline
 \end{tabular}
\caption{Detection using LSTM and MLP network }
\label{t2}
\end{table}
As can be seen from Table \ref{t2}, though the difference of detection performance using LSTM and MLP is insignificant, the training time using LSTM is much larger than MLP.

As is mentioned before, cyber attacks might be rare events. In this case, the dataset is not evenly distributed.
We consider 6 data composition scenarios where the data instances under normal load variation are set to the fixed value $100$ and the data instances under attack scenarios are changed from $60$ to $10$ with an interval of $10$; meanwhile, the rate of data instances for test is set to $70\%$ to $30\%$ with an interval of $10\%$. Accuracy rate is calculated to identify the probability of each instance's belonging to its right class. The results are shown in Fig. \ref{mlp1}
\begin{figure}[htbp]
\centering
\includegraphics[width=3.5 in]{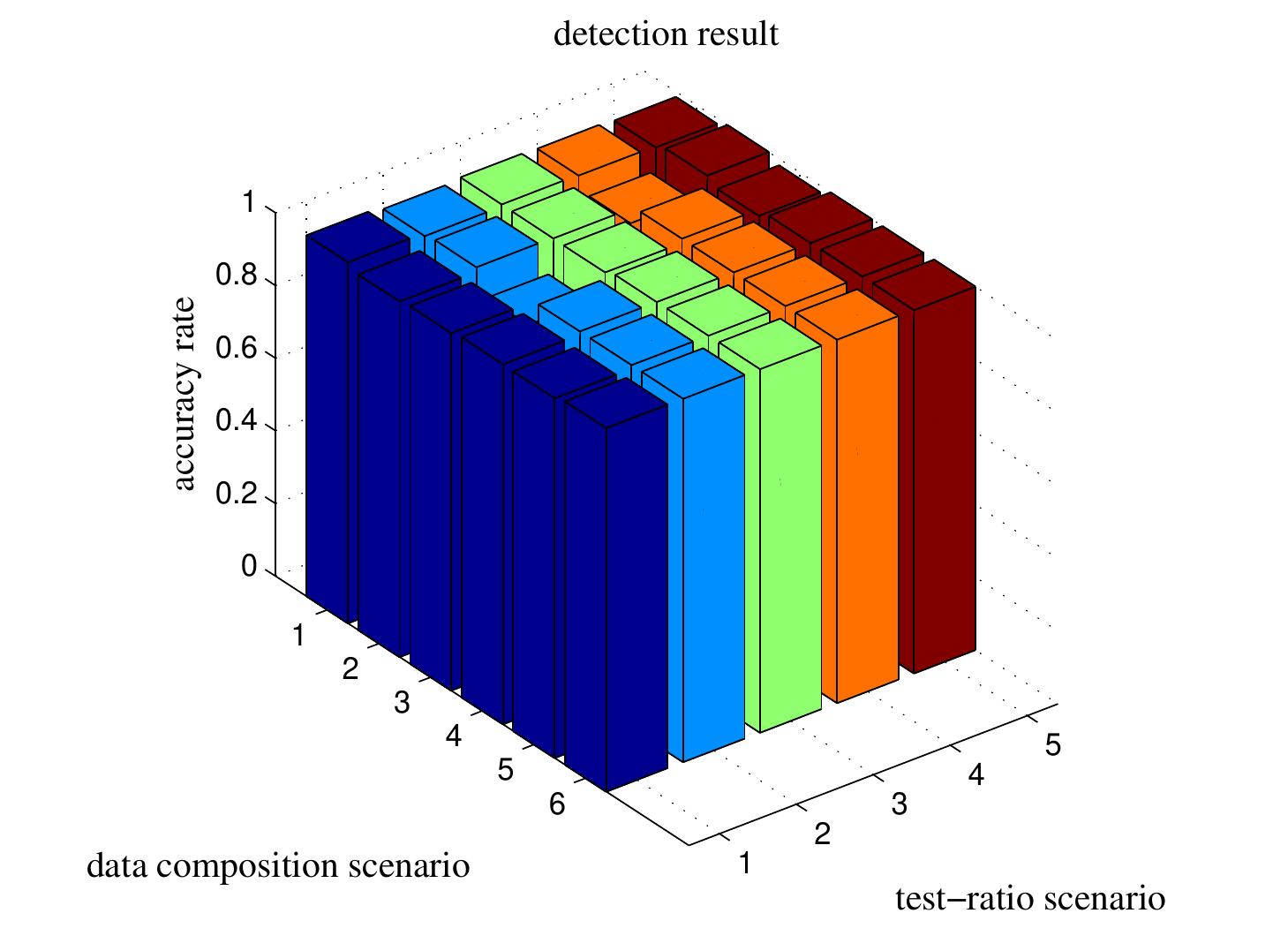}
\caption{Detection results under different data composition and test-ratio scenarios (MLP)}
\label{mlp1}
\end{figure}
In Fig. \ref{mlp1}, the number on $x$ axis represents the data composition scenario, e.g., $1$ indicates the data composition is: 100(normal load variation): 60(attack scenario 1): 60(attack scenario 2): 60(attack scenario 3); $2$ indicates the data composition is: 100(normal load variation): 50(attack scenario 1): 50(attack scenario 2): 50(attack scenario 3), and so forth. Similarly, the number on $y$ axis represents the test ratio scenario, e.g., $1$ indicates $70\%$ of the data is used for test while $30\%$ is used for training; $2$ indicates $60\%$ of the data is used for test while $40\%$ is used for training, and so forth. As can be seen in Fig. \ref{mlp1}, detection results (accuracy rate) under different training (testing) conditions are generally acceptable, even the 'worst' performance (under the condition where data composition is: 100:40:40:40 and test rate is $60\%$) is $0.925$. It can also be found that the scarcity of data instances from attack scenarios (100:10:10:10) does not significantly influence the detection performance, which is in compliance with actual detection conditions.

\subsubsection{Detection using other classifiers}
In this section, the classifiers based on Bayesian networks and SVM are tested.
The dataset (including data composition and test-ratio settings) are the same as Section \ref{simu21}.
The DFT coefficients are fed into Baye's networks and SVM, respectively. After learning the model, the test results are shown in Fig. \ref{gau} and \ref{svc}.
\begin{figure}[htbp]
\centering
\includegraphics[width=3.5 in]{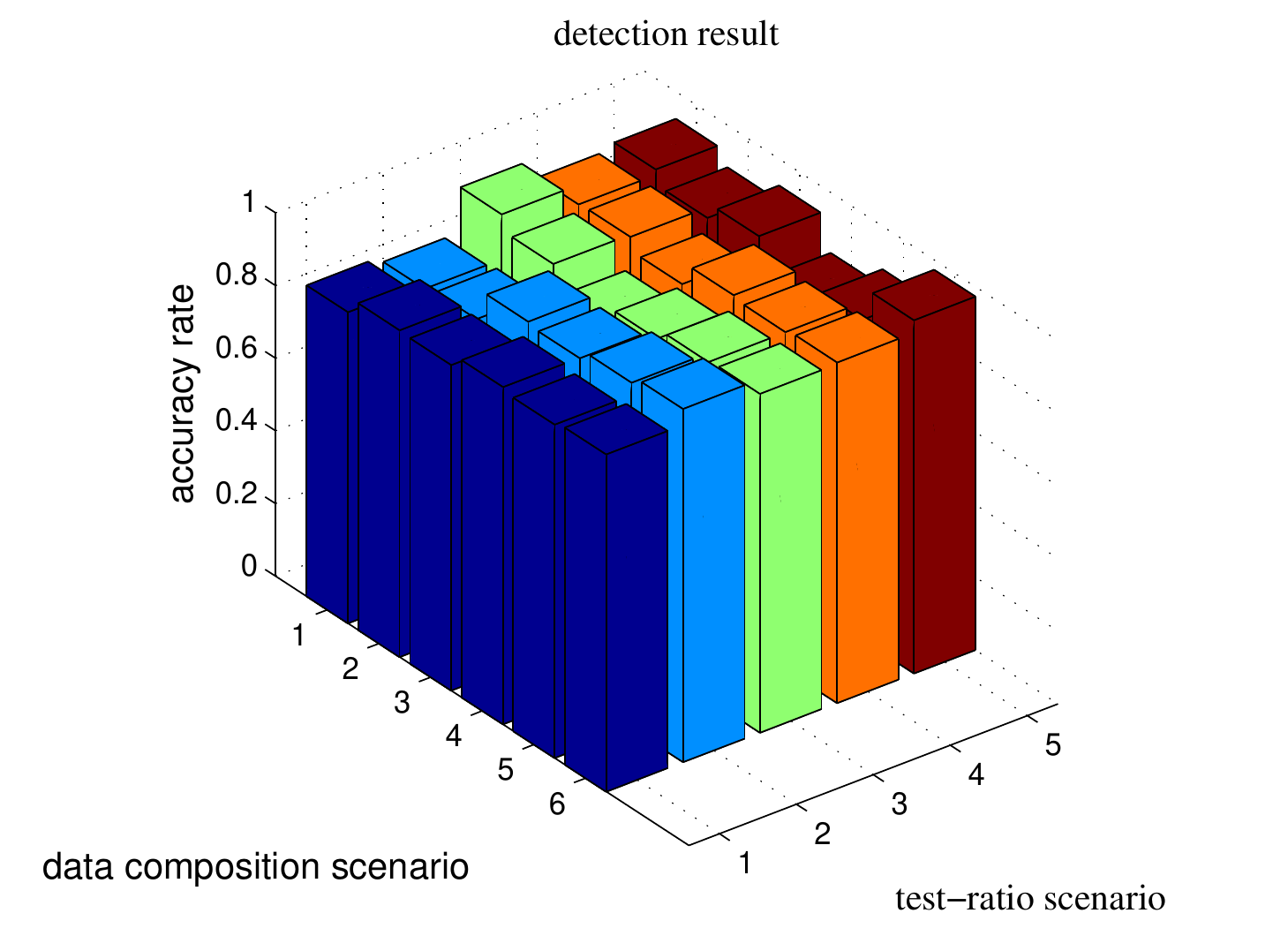}
\caption{Detection results under different data composition and test-ratio scenarios (Baye's network)}
\label{gau}
\end{figure}

\begin{figure}[htbp]
\centering
\includegraphics[width=3.5 in]{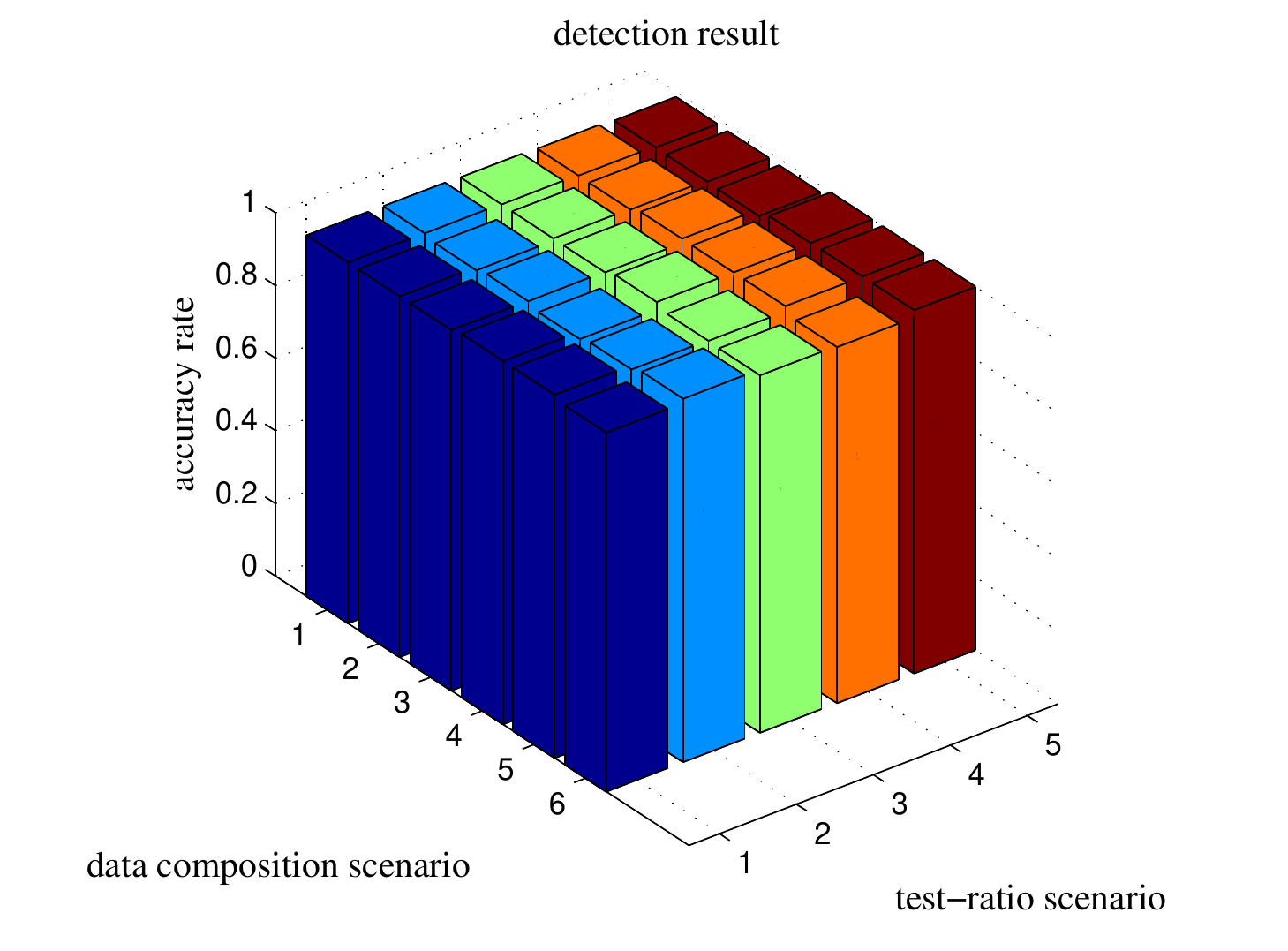}
\caption{Detection results under different data composition and test-ratio scenarios (SVM)}
\label{svc}
\end{figure}
Compared with Fig. \ref{mlp1} and \ref{svc}, it can be learned that identification performance using Bayesian networks is generally a little worse than MLP or SVM. Still, the 'worst' performance using Bayesian networks is acceptable (0.854). That is to say, the extracted feature (DFT coefficients) can effectively support the common basic classification algorithms for attack scenario identification, avoiding resorting to more complex and computationally inefficient learning models. The simulation results also show that training data quantities and the proportion of training (testing) data are permitted to vary in a wide range, verifying the robustness using the proposed feature extraction method.

\section{Conclusion}
In this paper, vulnerability for cyber security of LFC system is studied. Through the theoretical and numerical analyses in Section \ref{sec_att} and \ref{sec_cas}, it can be learned that system responses to attacks on different LFC components show varying severity degrees with frequency oriented attack producing the worst outcome. It means frequency related settings should be given priority in respect to protection. As for attack detection, it is learned that DFT coefficients under time-to-frequency-domain transform serve an effective feature for classification-based detection algorithm, which can identify different type of attack scenario and reduce computational burden.

\appendices
\section{Attack Scenarios in Section \ref{subsec_att}}
\label{appen1}
\begin{itemize}
  \item Scenario 1(2) injection (scale) attack on frequency measurement $f_{11}$ of Unit 1 of Area 1
  \item Scenario 3(4) injection (scale) attack on frequency measurement $f_{12}$ of Unit 2 of Area 1
  \item Scenario 5(6) injection (scale) attack on power measurement $P^{11}_{tie}$ of Tie-line 1 of Area 1
  \item Scenario 7(8) injection (scale) attack on power measurement $P^{12}_{tie}$ of Tie-line 2 of area 1
  \item Scenario 9(10) injection (scale)  attack on area frequency measurement  $f_1$
  \item Scenario 11 injection attack on nominal frequency  $f_0$
  \item Scenario 12(13) injection (scale) attack on area frequency deviation $\Delta f_1$ of Area 1
  \item Scenario 14(15) injection  attack on nominal power of Tie-line 1(Tie-line 2) $P^{11}_{tie,0}$ $P^{12}_{tie,0}$ of Area 1
  \item Scenario 16(17) injection (scale) attack on power deviation of Tie-line 1 $\Delta P^{11}_{tie}$ of Area 1
  \item Scenario 18(19) injection (scale) attack on power deviation of Tie-line 2 $\Delta P^{12}_{tie}$ of Area 1
   \item Scenario 20(21) injection (scale) attack on power deviation of Tie-line $\Delta P^{1}_{tie}$ of Area 1
    \item Scenario 22(23) injection (scale) attack on LFC order $comm_{11}$ to Unit 1 of Area 1
    \item Scenario 24(25) injection (scale) attack on LFC order $comm_{12}$ to Unit 2 of Area 1
\end{itemize}

\ifCLASSOPTIONcaptionsoff
  \newpage
\fi
\bibliographystyle{ieeetr}
\bibliography{sec}

\end{document}